\documentclass[11pt, letterpaper]{article}
\usepackage{paralist, fullpage,amsmath,amssymb, graphicx, algorithm2e}

\newtheorem{theorem}{Theorem}
\newtheorem{lemma}[theorem]{Lemma}
\newtheorem{definition}[theorem]{Definition}
\newtheorem{example}[theorem]{Example}
\newtheorem{proposition}[theorem]{Proposition}

\def\FullBox{\hbox{\vrule width 8pt height 8pt depth 0pt}}

\newcommand{\qed}{\;\;\;\FullBox}

\newenvironment{proof}{\noindent{\bf Proof:~~}}{\qed}

\newcommand{\bbF}{\mathbb{F}}

\newcommand{\supp}{\mathrm{supp}}
\newcommand{\sugg}{\leadsto}

\title{PAC Quasi-automatizability of Resolution over Restricted Distributions}
\author{Brendan Juba\thanks{Supported by ONR grant number N000141210358.}\\
Harvard University\\{\tt bjuba@alum.mit.edu}}

\begin{document}
\maketitle

\begin{abstract}
We consider principled alternatives to unsupervised learning in data mining by
situating the learning task in the context of the subsequent analysis task.
Specifically, we consider a query-answering (hypothesis-testing) task: In the 
combined task, we decide whether an input query formula is satisfied over a
background distribution by using input examples directly, rather than invoking
a two-stage process in which {\em (i)} rules over the distribution are learned 
by an unsupervised learning algorithm and {\em (ii)} a reasoning algorithm 
decides whether or not the query formula follows from the learned rules. In a 
previous work~\cite{juba13}, we observed that the learning task could satisfy 
numerous desirable criteria in this combined context -- effectively matching 
what could be achieved by agnostic learning of CNFs from partial information -- 
that are not known to be achievable directly. In this work, we show that 
likewise, there are reasoning tasks that are achievable in such a combined 
context that are not known to be achievable directly (and indeed, have
been seriously conjectured to be impossible, cf. Alekhnovich and Razborov~\cite
{ar08}). Namely, we test for a resolution proof of the query formula of a given 
size in quasipolynomial time (that is, {\em ``quasi-automatizing''} resolution).
The learning setting we consider is a partial-information, 
restricted-distribution setting that generalizes learning parities over the 
uniform distribution from partial information, another task that is known not to
be achievable directly in various models (cf. Ben-David and Dichterman~\cite
{bdd98} and Michael~\cite{michael10}).
\end{abstract}

\newpage

\section{Introduction}


When learning is employed in data mining and artificial intelligence, the
objective is not simply to form hypotheses that are supported by the data, but
to draw further inferences based on these learned hypotheses. The use of logical
inference on such learned hypotheses is problematic since learning algorithms
(namely, PAC-learning algorithms here~\cite{valiant84}) do not guarantee that
the hypotheses they produce can be interpreted as ``valid'' in the standard 
sense of mathematical logic. Motivated by this lacuna, Valiant~\cite{valiant00}
introduced {\em PAC-Semantics} (for a logic) to enable the analysis of logical 
inferences drawn from learned hypotheses. Such logical inferences are only
interesting in a partial-information model, since otherwise a query about a
candidate conclusion can be evaluated directly by evaluating the query formula 
on each complete example and considering the fraction of examples that satisfy 
the query.   

Valiant's work analyzed a two-stage process in which in the first stage, 
learning algorithms are used to extract explicit premises from examples drawn 
from a background distribution. (``Supervision'' is only provided in that the
rules predict individual bits of the examples based on others.) In the second 
stage, conclusions are drawn from these premises by applying logical inferences.
A stronger approach would be to ask, given a query formula as input, if there 
exist some premises that can be verified to hold on the (incomplete) examples, 
for which logical inference could derive the given query. The work of Khardon 
and Roth~\cite{kr97} on {\em learning to reason} suggests that it may be 
advantageous to take the two tasks together rather than in two separate stages 
as done by Valiant: they show that by processing the data directly, it is 
possible to answer queries (in a slightly different model) for which the 
reasoning problem is NP-hard, and for which the hypotheses that would support 
such queries are general DNFs, and hence not known to be learnable. In a 
follow-up work, Khardon and Roth~\cite{kr99} showed similar results for a 
answering a more limited class of queries from incomplete information.


One of Khardon and Roth's major objectives in their works was to provide 
tractable reasoning by avoiding the use of theorem-proving techniques. Indeed,
in most approaches to these problems of learning and reasoning that are actually
used in practice (e.g., based on graphical models) the resulting reasoning 
problem is known to be \#P-hard, even for highly restricted classes of formulas 
(cf. Roth~\cite{roth96}). Khardon and Roth's approach~\cite{kr99} could still 
only answer CNF queries under partial information that were, for example, Horn 
clauses or of low width. Our recent work~\cite{juba13} established that the 
combined problem of evaluating a query formula from partial examples using 
theorem-proving is actually no more difficult than the ``second stage'' problem 
of exact, classical reasoning for all ``natural'' (in the sense of Beame et 
al.~\cite{bks04}) fragments of proof systems.\footnote{%
We note some similarly motivating results of Alekhnovich et al.~\cite{abfkp08}, 
showing how theorem-proving algorithms can solve classical learning problems, 
and thus that learning is easier than theorem-proving over the same 
representations.} 
Such results raised the possibility that in some cases, such theorem-proving 
under PAC-Semantics may actually be easier than in the classical case. The 
results of the current work suggest that this is actually so: we 
quasi-automatize\footnote{%
{\em Automatization} is the standard terminology for the problem of deciding
whether or not a proof of a formula exists in a given (fragment of a) proof 
system. {\em Quasi-}automatization means that there is a quasipolynomial time 
algorithm for this problem.} the ({\em general, dag-like}) resolution proof 
system under PAC-Semantics (i.e., when given incomplete examples) with 
restricted distributions, whereas the best known automatization of resolution, 
first obtained by Clegg et al.~\cite{cei96}, achieves a running time of 
$n^{O(\sqrt{n\log n})}$ for polynomial-size resolution proofs over $n$ 
variables. Moreover, results by Alekhnovich and Razborov~\cite{ar08} suggest 
that perhaps resolution is not quasi-automatizable at all.


In this work (in contrast to \cite{juba13}), we consider problems in which the
background distribution (and hence learning problem) is restricted. The main
example we consider is a generalization of learning parities under the uniform
distribution: we note that one way of viewing the (standard) problem of learning
parities over $n$ bits is that there is an unknown parity constraint over the 
$n+1$ bits; here, we consider an ``unsupervised'' setting where there is again 
no distinguished label bit, and there may be several such constraints. We refer 
to such distributions as {\em affine distributions.} Such problems may easily be
impossible under partial information: in the (weakly) Restricted Focus of 
Attention (wRFA) model, Ben-David and Dichterman~\cite{bdd98} showed that 
learning a large parity becomes impossible if an insufficient number of bits can
be simultaneously viewed. We remark that such parity formulas are also natural 
and interesting from the standpoint of the reasoning problem, since they 
presented classic examples of theorems that were hard for resolution~\cite
{tseitin70}.

Our theorems are actually established for a more general class of distributions 
that merely feature a {\em ``correlation gap''}: when we condition on a 
conjunction, the other bits are either highly biased or feature bounded bias 
(with the bounds on these biases parameterizing the ``gap'' that gives these 
distributions their name). Such distributions are reasonably natural. For
example, some simple topic models, e.g., along the lines of the {\em pure 
document} model underlying the analysis of Latent Semantic Indexing given by 
Papadimitriou et al.~\cite{prtv00}, feature such correlation gaps. Roughly,
the presence or absence of ``primary terms'' in a document determine the 
relative prevalence or absence of other primary terms by determining the latent 
topic variable, and other words do not have such an effect.

\subsection{A closer look at our results}


To get a sense of how automatization under PAC-Semantics could possibly be 
easier than the usual, ``worst-case'' automatization, we first consider the case
of automatizing over the uniform distribution in Section~\ref{unif}. We 
observe that clauses of at least logarithmic width have a satisfied literal -- 
i.e., are {\em witnessed satisfied} -- in our partial examples with high 
probability. Using the techniques of our previous work~\cite{juba13}, it 
therefore suffices to search for proofs of small width for each of the examples,
which is known to be automatizable. 

We then turn to the learning context, in which there are parity constraints on 
the distribution that could be useful as premises in a small proof of the query.
Our algorithm distinguishes when an input CNF can be refuted by a small 
resolution proof on the basis of some learnable formulas over the unknown 
distribution (e.g., some small parity constraints) from when the input CNF is 
satisfied with moderate probability, that is, distinguishing when its negation
(a DNF) can be proved by a small resolution proof, from when its negation is
falsified with moderately large probability:

\begin{theorem}[Main theorem, cf. Theorem~\ref{mainthm}]\label{mainthm-intro}
There is an algorithm that,
given an input CNF $\varphi$, a bound $p(n)$, $\mu,\epsilon,\gamma,\delta\in [0,
1]$, and access to partial examples from an affine distribution $D$ in which
indices are hidden independently with probability $1-\mu$ in each example, and 
given that either
\begin{compactitem}
\item $\varphi$ is satisfied by $D$ with probability at least $\epsilon+\gamma$ 
or
\item there is a resolution refutation of size $p(n)$ of 
$\psi_0\wedge\psi_1\wedge\varphi$ for CNFs $\psi_0$ and $\psi_1$ such that 
\begin{compactitem}
\item $\psi_0$ consists of $1-\frac{\gamma}{n^{O(\log\frac{p(n)}
{\gamma})}}$-valid clauses and
\item $\psi_1$ is witnessed to evaluate to true with probability 
$1-\epsilon+\gamma$ on the partial examples
\end{compactitem}
\end{compactitem}
decides which case holds with probability $1-\delta$, and runs in time 
$n^{O(\frac{1}{\mu}\log\frac{p(n)}{\gamma\delta})}$.
\end{theorem}

Our main theorem supports the learning of two kinds of hypotheses, denoted 
$\psi_0$ and $\psi_1$ in the theorem statement. The conditions on $\psi_0$ 
correspond essentially to standard, ``realizable'' learning. Although in 
contrast to the usual learning set-up we allow some potentially noticeable 
counterexamples to the target hypothesis, these counterexamples are required to 
be very rare relative to the accuracy parameter $\gamma$ given to the algorithm.
Since the parity constraints in an affine distribution hold perfectly, the 
learnability of such a $\psi_0$ establishes the learnability of the parity 
constraints of an affine distribution in the context of a resolution proof, 
which after all can only operate on clauses. We remark that our previous 
work~\cite{juba13} could not exploit such hypotheses in general, due to the
fact that we did not restrict our attention to the independent masking 
processes.

By contrast, the conditions required in our main theorem on the second 
hypothesis $\psi_1$ allow it to have an arbitrary error tolerance $\epsilon$
given as an input parameter to the algorithm. In this case, our algorithm is
essentially detecting when there exist any CNF formulas that suffice to complete
a proof of the query, that are $\epsilon$-close to being valid. For example, in
the simple topic model distributions, clauses encoding a rule that the presence
of one primary term implies that at least one out of $\Omega(\log 1/\epsilon)$ 
of the other primary terms for the same topic also appears is a $(1-
\epsilon)$-valid clause that might be useful in reasoning about such documents. 
The permission of an arbitrary $\epsilon$ flaw in the formulas here may be seen 
as analogous to the {\em agnostic learning} setting in traditional concept 
learning~\cite{kss94} and we feel that it is crucial to the utility of such 
algorithms in artificial intelligence, for example. The conditions we impose on 
this second hypothesis follow the conceptual developments from our previous 
work~\cite{juba13} and this groundwork merely serves as a starting point in the 
present setting. Still, for the reader's sake, we will briefly review why we 
deem the restrictions imposed on this hypothesis to be reasonable and how they 
lead to a tractable learning problem. Agnostic learning is notoriously hard, and
the main reason we avoid this problem is by declining to specify {\em any} 
explicit representation of a relation that is $\epsilon$-close to valid,%
\footnote{Surely, this is a step beyond merely specifying another explicit
representation that lies outside the target concept class as done in improper 
learning, and hence we count it as a separate issue.}
a problem that is conveniently irrelevant to our ultimate goal of answering a 
query here. The more serious difficulty that impacts the class of hypotheses we
ultimately ask to learn is that determining whether or not a formula is 
satisfied in the context of partial information can be intractable or even
impossible, depending on the level of generality; in our previous work
we observed that this leads to a serious obstacle since the ability to detect 
when queries are provable implies that we can detect that the premises are 
satisfied. As this may not be feasible, it is therefore unreasonable
to hope for an algorithm that utilizes {\em any} CNF that is $\epsilon$-close to
being valid as a hypothesis; what we {\em can} hope for instead (and is achieved
by the previous work) is that we can utilize hypotheses that could be 
efficiently verified to be $\epsilon$-close to being valid if only we had 
identified such a hypothesis. (Note that even in the usual complete information 
agnostic learning setting, as hard as it is, this condition is always 
satisfied.) The solution proposed in the previous work was to use a simple 
definition of {\em witnessed evaluation} that suffices to enable efficient 
certification that a formula is satisfied in a partial example. As illustrated 
in our warm-up, this definition moreover turns out to be strong enough to enable
solutions to the combined learning and reasoning problem.

The core of the proof of the main theorem is a structural result for resolution
proofs over correlation gap distributions: every formula with a polynomial-size 
resolution refutation (possibly using clauses that hold over the background 
distribution as additional premises) simplifies to a logarithmic-width 
refutation using logarithmic-width clauses that are learnable from the 
background distribution. Essentially, we only need to focus on clauses in the 
proof that cannot be learned, and the key observation is that the only way a 
clause of sufficient (logarithmic) width could not be witnessable (and so 
unlearnable) is if many of its literals are biased to be false; this may happen 
if the literals of the clause appear in some even-parity constraint, for 
example. There is then some other (small) learnable clause that can be used to 
eliminate each such literal, and so by a sequence of resolution steps, we can 
reduce the wide clauses of the proof down to logarithmic width. We similarly 
establish the learnability of the nearly-realizable $\psi_0$ by showing that we 
only need to consider logarithmic-width subclauses for the proof. These 
techniques actually also allow us to determine whether or not a given CNF query 
is highly valid directly: In Section~\ref{perfect-cnf}, we exhibit an algorithm 
that distinguishes when an input CNF is satisfied with high probability over the
background distribution from when it is falsified with some moderate 
probability, without reference to resolution proofs. This latter result may be 
understood as lifting the restrictions on the CNF query in Khardon and Roth's 
work~\cite{kr99}, given some new restrictions on the distribution and obscuring 
of information.\footnote{%
By contrast, our main theorem essentially certifies the validity of DNF queries,
which seems to require the existence of simple proofs in the context of 
incomplete information---a further discussion of this point appears in 
Section~\ref{pac-res-bg}, in the context of the formal set-up.} 
These techniques for learning to reason about CNFs from partial information are 
new to this work, and seem to rely on both the correlation gap and the relative 
simplicity of the process that hides information.

\section{Our setting: background and approach}\label{pac-res-bg}

\paragraph{PAC-Semantics.}
PAC-Semantics (for a logic) was introduced by Valiant~\cite{valiant00} to
capture the kind of validity possessed by statements produced by PAC-learning
algorithms; for example, if one uses a PAC-learning algorithm to learn a 
conjunction $\wedge_j x_{i_j}$ to match the ``labels'' given by another variable
$x_t$ from examples from a distribution $D$, then the formula $[\wedge_jx_{i_j}
\equiv x_t]$ is (probably) approximately valid in the following sense:

\begin{definition}[$(1-\epsilon)$-valid]
Given a distribution $D$ over $\{0,1\}^n$, we say that a Boolean formula 
$\varphi$ is {\em $(1-\epsilon)$-valid} if 
$\Pr_{x\in D}[\varphi(x)=1]\geq 1-\epsilon$. If $\epsilon=0$, we say that the
formula is {\em perfectly valid}.
\end{definition}

Of course, the definition makes sense for other kinds of formulas (not just
equivalences). It is not hard to show (by a union bound) that any classical 
logical inference can be applied to formulas possessing this weaker kind of 
validity, as long as we allow for further loss in the approximation.\footnote{%
It also is not hard to show that as long as the distributions are arbitrary, the
union bound is tight here~\cite{juba13}.}

\begin{proposition}[Classical reasoning in PAC-Semantics~\cite{juba13}]
\label{classical-inf-bound}
Let $\psi_1,\ldots,\psi_k$ be formulas such that each $\psi_i$ is
$(1-\epsilon_i)$-valid under a common distribution $D$ for some $\epsilon_i\in
[0,1]$. Suppose that $\{\psi_1,\ldots,\psi_k\}\models\varphi$ (in the classical
sense). Then $\varphi$ is $1-\epsilon'$-valid under $D$ for
$\epsilon'=\sum_i\epsilon_i$.
\end{proposition}

The main problem that we wish to address, introduced in prior work~\cite
{juba13}, is that of deciding the degree of validity of a given {\em query}
formula using {\em examples} from an unknown distribution. In the present work, 
we will see how to obtain stronger results (than \cite{juba13} in particular) 
for a moderately restricted class of distributions, in which the correlation 
between variables is either strong or weak (and not of moderate strength):

\begin{definition}[Correlation gap]
We will say that a distribution $D$ over $\{0,1\}^n$ has a {\em width-$w$
$(\beta,1-\gamma)$ correlation gap} if for any conjunction of $k\leq w$ literals
$\ell_1\wedge\ldots\wedge\ell_k$ satisfied with nonzero probability and any 
variable $x$,
either $\Pr[x=1|\ell_1\wedge\cdots\wedge\ell_k]\geq\beta$ and $\Pr[x=0|\ell_1
\wedge\cdots\wedge\ell_k]\geq\beta$ or else for some $b\in\{0,1\}$, $\Pr[x=b|
\ell_1\wedge\cdots\wedge\ell_k]\geq 1-\gamma$. In the former case, we say that 
$x$ is {\em $\beta$-balanced for $\ell_1\wedge\cdots\wedge\ell_k$,} 
and in the latter we say that, respectively, $x$ (for $b=
1$) or $\neg x$ (for $b=0$) is {\em $1-\gamma$-implied by $\ell_1\wedge\cdots
\wedge\ell_k$,} denoted as $\ell_1\wedge\cdots\wedge\ell_k\sugg x$ (or $\ell_1
\wedge\cdots\wedge\ell_k\sugg\neg x$, respectively).
\end{definition}

A simple example of a large collection of distributions that have a very strong 
correlation gap of $(1/2,1)$ up to width $n$ and serve as a motivating example, 
is the class of {\em affine distributions}:

\begin{definition}[Affine distribution]
For any solvable linear system over $\bbF_2$ $Ax=b$, the distribution over
$\{0,1\}^n$ that is uniform over solutions to the linear system is
an {\em affine distribution.}
\end{definition}

Affine distributions turn out to have a correlation gap, since intuitively,
if a conjunction determines the value of another variable in a linear 
constraint, that literal is easily seen to be $1$-implied, and otherwise the
literal turns out to be uniformly distributed (and thus, $1/2$-balanced). 

\begin{definition}[Constraints on clauses]
We say that there is a {\em constraint} on a clause $C$ in an affine 
distribution given by the linear system $Ax=b$ if there is a linear combination
of the rows of $A$ such that the only nonzero entries are in indices $i$ for
which the corresponding variable appears in a literal of $C$. 
\end{definition}

\begin{lemma}
\label{noconsts}
Let $C$ be a clause and $D$ be an affine distribution such that there are no
constraints on $C$. Then the marginal distribution over the variables appearing
in $C$ is uniform.
\end{lemma}
\begin{proof}
Let $Ax=b$ be the linear system defining $D$. Let $y$ denote the variables in 
$C$ and $z$ denote the variables not appearing in $C$, and let $A'$ be the
submatrix of $A$ formed by columns corresponding to variables in $C$, and $A''$ 
be the submatrix formed by columns of $A$ corresponding to variables not in $C$.
Note that by converting the matrix into reduced form, we may verify that if for
some $y^*$ there were no solutions to the system $A''z=(b+A'y^*)$, then there
must be some linear combination of the rows of $A''$ yielding the $0$ vector.
The corresponding linear combination of the rows of the linear system $Ax=b$
would be a constraint on $C$, which does not exist by hypothesis. Therefore,
for every $y^*$, the system $A''z=(b+A'y^*)$ has solutions, and furthermore,
since $A''$ is the same for every assignment $y^*$, it always has the same 
number of solutions. Thus, as $D$ is uniform over these solutions (and $y^*$),
the marginal distribution over the variables in $C$ is indeed uniform as 
claimed.
\end{proof}

Correlation gap distributions are reasonably natural. For example, the
simple {\em ``pure document''} probabilistic corpus model introduced by
Papadimitriou et al.~\cite{prtv00} to analyze Latent Semantic Indexing generally
features a nontrivial correlation gap:

\begin{example}[Pure document topic model~\cite{prtv00}]
The {\em pure document topic model} is a probabilistic model of document
generation. We suppose that documents are represented by the set of words
appearing in them. A document is then generated in a two-stage process in which
a (latent) {\em topic variable} is first sampled, where each outcome of the
topic variable is associated with one member of a family of disjoint sets of 
{\em primary terms}. The set of words associated with this topic is obtained by
the union of these primary terms with a set of generic terms (shared across 
topics); each word has an associated probability in the range $[\epsilon,\tau]$ 
for some small constants $\epsilon$ and $\tau$. The document itself is then 
sampled by independently tossing a biased coin for each word in the topic set, 
including it with its given probability. If the overall probability of the topic
words appearing in the document is at most $\delta\ll\epsilon$, which holds, for
example, if each topic has sufficiently small probability of being chosen 
(relative to $\epsilon/\tau$), then the distribution has a width-$w$ $(\epsilon,
1-\delta(1+\delta w))$-correlation gap.
\end{example}

\paragraph{Partial assignments.}
Answering queries using (complete) examples drawn from a distribution $D$ is
trivially accomplished. By contrast, the task is much more difficult when some
of the variables' settings are deleted from the examples, resulting in 
{\em partial} example assignments. In the present work, we will focus on a very
simple process for generating partial assignments that chooses whether to delete
each entry from the assignment by tossing a biased coin. In learning theory,
such a model first appeared in the work of Decatur and Gennaro~\cite{dg95}.
Formally:

\begin{definition}[Partial assignments and masking]
A {\em partial assignment} $\rho$ is an element of $\{0,1,*\}^n$. We say that a
partial assignment $\rho$ is {\em consistent} with an assignment $x$ if whenever
$\rho_i\neq *$, $\rho_i=x_i$.

A {\em mask} is a function taking assignments to consistent partial assignments.
A {\em masking process} is a mask-valued random variable. We denote the 
distribution obtained by applying a masking process $M$ to a distribution over
assignments $D$ by $M(D)$.

The {\em independent masking process with parameter $\mu$} is the following
masking process $M_\mu$: for every $x\in\{0,1\}^n$, for every $i$, $M_\mu(x)_i=
x_i$ independently with probability $\mu$, and otherwise $M_\mu(x)_i=*$.
\end{definition}


We note that a related set-up appears in recent works by Dvir et al.~\cite
{drwy12}, Wigderson and Yehudayoff~\cite{wy12}, and Moitra and Saks~\cite{ms13} 
with the latter result(s) being particularly relevant: in our terms, they 
essentially show that there is a polynomial time algorithm (in $|\supp(D)|$,
$\frac{1}{\epsilon},\log\frac{1}{\delta},$ and $n$, for constant $\mu\in (0,1]$)
that for any distribution $D$ on $\{0,1\}^n$, completely recovers $D$ to within 
an additive $\epsilon$ at each point in $\supp(D)$ given access to $M_\mu(D)$. 
As a consequence, they obtain (in particular) an algorithm for answering queries
that is efficient when $D$ has sparse support: we simply recover the 
distribution (to within an additive $\epsilon=\gamma/|\supp(D)|$ error, say) and
add up the mass assigned to points where the formula is true to obtain an 
additive $\gamma$ estimate of the degree of validity of the query under $D$.

We will be interested in evaluating queries for distributions that have 
exponentially large support, such as affine distributions. In order to do
so, we will need to impose further restrictions on the formula, namely (looking
ahead) that the formula has a small resolution proof from clauses that are 
either always true or merely {\em witnessed} to be true on most examples from
$M_\mu(D)$:

\begin{definition}[Witnessing]
We define a clause to be {\em witnessed to evaluate to true} on a partial 
assignment if there is some literal $\ell(x_i)$ ($x_i$ or $\neg x_i$) in the 
clause such that ($\rho_i\neq *$ and) $\ell(\rho_i)=1$. We say that the clause 
is {\em witnessed to evaluate to false} if for every literal $\ell(x_i)$ in the 
clause, ($\rho_i\neq *$ and) $\ell(\rho_i)=0$.
\end{definition}

A very natural and related notion is using the partial assignment as a 
{\em restriction} to simplify a formula by ``plugging in the known variables.''
For the cases of interest (CNFs here) this operation can be simply captured by 
the following definition.

\begin{definition}[Restricted formula]
Given a partial assignment $\rho$ and a clause $C$, the {\em restriction of $C$
under $\rho$,} denoted $C|_\rho$, is the tautological formula $\top$ if $C$
is witnessed to evaluate to true on $\rho$, and otherwise it is given by the
set of literals $\ell(x_i)$ in $C$ for which $\rho_i=*$. For a CNF $\varphi$, if
some clause is witnessed to evaluate to false, then the restriction of $\varphi$
under $\rho$ is defined to be the contradictory formula $\bot$; otherwise, it is
the conjunction of restrictions of clauses not witnessed to evaluate to true.
\end{definition}



\paragraph{Theorem-proving versus other approaches to answering queries.}
When a (small) resolution proof exists, this implies that the query has the form
of a DNF formula. This naturally raises the question of whether techniques 
similar to those used by Wigderson and Yehudayoff~\cite{wy12} (in particular) 
could be applied to the partial assignments corresponding to the DNF's terms to 
evaluate such a DNF directly, without turning to a resolution proof. The 
difficulty that arises with such an approach is that when many DNF terms 
correspond to partial assignments that are consistent with one another, 
determining the weight of their disjunction by such techniques seems to involve 
an (exponentially large) inclusion-exclusion expression.\footnote{%
Wigderson and Yehudayoff avoid the need for inclusion-exclusion by only
considering complete assignments, which of course correspond to disjoint
events.}
While one could try to bound the number of terms in the sum by considering only 
terms of bounded width using the techniques we develop here (cf. Lemmas \ref
{wide-witnessed} and \ref{consconsts} and the approach of Theorem~\ref
{mainthm}), decreasing the error per omitted term by increasing the width bound 
increases the number of terms at a greater rate, so that a nontrivial bound on 
the error cannot be obtained. By contrast, the existence of a small a priori 
bound on the size of a resolution proof of the formula enables us to derive a 
related bound on the width we need to consider (naturally, this is made precise 
in the proof of Theorem~\ref{mainthm}).

Khardon and Roth~\cite{kr99}, in their work on ``learning to reason,'' had a
different approach to an essentially similar problem---among others, they 
considered a problem where one is given a distribution over partial assignments
that are consistent with satisfying assignments of some background formula
(generally a polynomial-size DNF), and they wish to decide whether the query 
formula is either entailed by the background formula or whether it is falsified 
with some given probability (promised that one of these two cases holds). When
the query is a $k$-CNF (say for $k=O(\log n)$), they answer the query by
constructing a relatively small set of partial assignments and testing whether
the query has an unsatisfied clause on any of the partial assignments; this
approach again answers queries while avoiding any consideration of 
theorem-proving. The size of the set (and therefore also the running time) is 
polynomial in $n$ and exponential in $k$, and so it is quasipolynomial for $k=
O(\log n)$. It turns out that our work can be used to extend their scheme to 
answer general CNF queries under distributions with strong correlation gaps and 
independent masking by first reducing the width of the (important) clauses in 
the query to $O(\log n)$, as we show in Section~\ref{perfect-cnf}. The reader 
should note that this approach only works for testing whether the query is {\em 
almost perfectly} valid.\footnote{%
Although the masking process is simple enough to allow us to compute 
(e.g.) unbiased estimates of the probability that a given (set of) clause(s) is
witnessed unsatisfied, recovering an estimate of the true probability that the
formula is unsatisfied seems to again involve a potentially exponentially large 
inclusion-exclusion calculation.}

We stress that these two approaches apply to largely orthogonal classes of 
formulas: the latter approach answers CNF queries, whereas the former, 
resolution-based approach answers DNF queries. For both approaches, there is an 
important asymmetry: in the first approach, it is the existence of a small 
resolution proof that restricts the class of formulas that can be used as 
hypotheses and queries, whereas for the second, it is a matter of the asymmetric
way the error arises when testing the query on a partial assignment.

\paragraph{Resolution.}
We now briefly review the resolution proof system, which plays a central role in
this work. Resolution is a proof system for establishing the {\em validity} of a
{\em DNF}: it is given by a {\em refutation} of its (CNF) negation.

\begin{definition}[Resolution]
A {\em resolution refutation} of a CNF $\varphi$ is given by a sequence of 
clauses $\{C_i\}_{i=1}^k$ such that $C_k$ is the (unsatisfiable) empty clause
$\bot$ and for each $C_i$ in the proof, one of the following holds:
\begin{compactenum}
\item {\em (Axiom)} $C_i$ is a clause of $\varphi$
\item {\em (Weakening)} some $C_j$ for $j<i$ is a subclause of $C_i$
\item {\em (Cut)} there exist clauses $C_j=D\vee x$ and $C_\ell=E\vee\neg x$ 
(for some variable $x$) with $j,\ell<i$ such that $C_i=D\vee E$. The literals
$x$ and $\neg x$ are said to be a {\em complementary pair}.
\end{compactenum}
\end{definition}

Note that resolution is sound: each $C_i$ derived in a step of the proof is a
consequence of $\varphi\wedge C_1\wedge\cdots\wedge C_{i-1}$. Thus, the 
derivation of the empty clause implies that $\varphi$ is unsatisfiable.
It is standard to note that weakening is unnecessary in a resolution refutation:
if we simulate a resolution proof without using the weakening rule (applying the
cut rule to the original clauses instead of weakened ones at each step) we end
up with a legal derivation that has a subclause of the original derivation at 
each step, which must also end with the empty clause. 
We include weakening as a derivation rule because it is useful in the analysis
of resolution. In particular, we will be interested in the result of ``plugging
in'' a partial assignment to each step of a resolution refutation:

\begin{definition}[Restricted proof]
Given a resolution refutation $\Pi$ and partial assignment $\rho$, the 
{\em restriction} of $\Pi$ under $\rho$, denoted $\Pi|_\rho$, is the proof
obtained by substituting $C|_\rho$ for each clause $C$ appearing in $\Pi$.
\end{definition}

Using the weakening rule, the following (folklore) observation is easily 
established.

\begin{proposition}[Resolution is restriction-closed]
\label{res-closed}
For any CNF $\varphi$ with a resolution refutation $\Pi$ and any 
partial assignment $\rho$, the restriction $\Pi|_\rho$ is a resolution 
refutation of $\varphi|_\rho$.
\end{proposition}

We will see the value of Proposition~\ref{res-closed} illustrated simply in 
Section~\ref{unif}. (It will, of course, also play a key role in establishing
the main theorem, Theorem~\ref{mainthm}.)

\section{Automatizability of resolution}\label{gap-proof}

\subsection{Width-based automatizability for the uniform distribution}
\label{unif}

We start with the simpler special case of the uniform distribution (which we
denote $U_n$), which will introduce a useful lemma for the general case, and 
help develop some intuition. The key observation here is that clauses that are 
sufficiently wide are almost always witnessed to evaluate to true:

\begin{lemma}
\label{wide-witnessed}
Let $C$ be a clause such that for any literal $\ell$ of $C$ and any 
subclause $C'$ of $C$ without $\ell$, $\ell$ is $\beta$-balanced for $\neg C'$, 
and suppose $C$ has width at least $\frac{1}{\mu\beta}\ln\frac{1}{\delta}$. 
Then $C$ is witnessed to evaluate to true on $M_\mu(D)$ with probability $1-
\delta$.
\end{lemma}
\begin{proof}
We note that $\rho$ drawn from $M_\mu(D)$ are of the form $\rho=m(x)$ for $m$
drawn from $M_\mu$ and $x$ drawn from $D$ independently. Suppose we construct
a sequence of literals and subclauses of $C$ as we sample $m$ and $x$ in the 
following way: put $C_0=\bot$, and fix the entries of $m$ in order until we 
encounter some unmasked entry corresponding to some literal $\ell_i$ in $C$; we 
then put $\ell_i$ in $C_{i+1}=C_i\vee\ell$. Each literal of $C$ is thus included
in some $C_i$ independently with probability $\mu$. Now, $C$ is witnessed true 
on $\rho$ precisely when some $\ell_i$ is set to true in $x$, and if the first 
$i-1$ literals are set to false in $x$, then since $\ell_i$ is $\beta$-balanced 
for $\neg C_i$, each $\ell_i$ is set to true in $x$ with probability at
least $\beta$ when the first $i-1$ are set to false. Thus, the probability that 
none of these literals in a clause of width $w$ is satisfied by $\rho$ is at 
most $(1-\mu\beta)^w$, which for $w\geq\frac{1}{\mu\beta}\ln\frac{1}{\delta}$ is
at most $\delta$.
\end{proof}

Since the uniform distribution is $1/2$-balanced, Lemma~\ref{wide-witnessed} 
establishes that any wide clauses that appear in a resolution refutation are 
witnessed to evaluate to true with high probability over $M_\mu(U_n)$, and hence
in a proof $\Pi$ of size $p(n)$, a union bound establishes that every clause of 
width $\Omega(\frac{1}{\mu}\log\frac{p(n)}{\delta}))$ is substituted by $\top$ 
in the restriction $\Pi|_\rho$. Thus, the width-based algorithm first studied
by Galil~\cite{galil77} that searches for refutations using clauses of width at 
most $w$ (and thus runs in time $O(n^{2w})$) finds a refutation when one exists
with probability $1-\delta$. The following theorem is then almost immediate:

\begin{theorem}\label{unifthm}
There is a quasipolynomial time algorithm that, given an input CNF $\varphi$,
size bound $p(n)$, $\mu,\epsilon,\gamma\in [0,1]$, and access to examples from 
$M_\mu(U_n)$, and given that either
\begin{compactitem}
\item $\varphi$ is satisfied by $U_n$ with probability at least $\epsilon+
\gamma$ or
\item there exists a CNF $\psi$ that is witnessed to evaluate to true with
probability at least $1-\epsilon+\gamma$ under $M_\mu(U_n)$ such that there is a
resolution refutation of size $p(n)$ of $\psi\wedge\varphi$
\end{compactitem}
decides which case holds with probability $1-\delta$.
\end{theorem}
\begin{proof}
The algorithm takes a sample of partial assignments of size $m=\frac{1}
{\gamma^2}\log\frac{1}{\delta}$ and for each partial assignment $\rho^{(i)}$ 
uses the width-based algorithm for $w=O(\frac{1}{\mu}\log\frac{m\cdot p(n)}
{\delta})$ to check for a width-$w$ refutation of $\varphi|_{\rho^{(i)}}$; if 
more than an $\epsilon$ fraction of these refutations fail, the algorithm 
rejects and otherwise it accepts. It is immediate that (for constant $\mu$) the 
algorithm runs in quasipolynomial time in $n$, $1/\gamma$, $\log 1/\epsilon$, 
and $1/\delta$. We thus turn to considering correctness.

If $\varphi$ is $\epsilon+\gamma$ valid, then Hoeffding's inequality shows
that with probability $1-\delta$, at least an $\epsilon$ fraction of the
assignments drawn from $U_n$ satisfy $\varphi$; since this satisfying assignment
$x^{(i)}$ is consistent with any partial assignment $\rho^{(i)}$ in the support
of $M_\mu(x^{(i)})$, $\varphi|_{\rho^{(i)}}$ is satisfiable, and hence by the 
soundness of resolution, no refutation exists for these partial assignments,
and we see that the algorithm rejects.

In the second case, we note first that (again by Hoeffding's inequality) with 
probability $1-\delta/2$, every clause of the unknown formula $\psi$ is 
witnessed to evaluate to true in at least a $1-\epsilon$-fraction of the
partial assignments. Let $\Pi$ be the size $p(n)$ refutation of $\psi\wedge
\varphi$. It then follows from Lemma~\ref{wide-witnessed} (cf. the 
above discussion) that with probability $1-\delta/2$, for this $1-\epsilon$
fraction of the $\rho^{(i)}$ (out of $m$), $\Pi|_{\rho^{(i)}}$ is a width-$w$ 
refutation of $\varphi|_{\rho^{(i)}}$ since the clauses from $\psi$ and clauses
of $\Pi$ of width greater than $w$ all simplify to $\top$. Thus, in this case, 
the algorithm accepts with probability $1-\delta$, as needed.
\end{proof}

\subsection{Augmenting the width-based algorithm with learning}
While the uniform distribution illustrates how the {\em reasoning} problem may 
become easier in the context of a distribution, in a sense there is no 
{\em learning} problem if the distribution is known to be uniform: a given
formula indicates which settings of the variables are ``positive'' or 
``negative'' examples, and the entire learning question for a given formula 
merely concerns whether the distribution assigns high weight to the negative 
examples. We now turn to considering our learning problem.

Distributions with a correlation gap turn out to be easy to work with because we
only need to consider narrow clauses: for starters, Lemma~\ref{wide-witnessed} 
guarantees that (sub)clauses of sufficient width ($\Omega(\frac{1}{\mu\beta}\log
\frac{1}{\delta})$ here) for which every variable is balanced in $D$ (for every
small subset of the rest) are witnessed with high probability. Naturally, in an
affine distribution, we can make a similar claim if all of the literals of a 
sufficiently wide clause are each involved in some constraint that $1$-implies 
one of them. More generally:

\begin{lemma}
\label{consconsts}
Let $C$ be a clause of width $2w$ for $w\geq\frac{1}{\mu}\ln\frac{1}{\delta}$ 
and $D$ be a distribution with a width-$w$ $(\beta,1-\gamma)$ correlation gap 
such that for every subclause of $C$ of size $w$ there is some further subclause
$C'\vee\ell$ with $\neg C'\sugg\ell$. Then with probability $1-\delta$, there is
an unmasked literal of $C$ that is true with probability $(1-\gamma)$ under $D$.
\end{lemma}
\begin{proof}
Let $C_0$ be the first $w$ literals of $C$, and suppose $m$ is drawn from
$M_\mu$ and $x$ is drawn from $D$. Given $C_i$ (of width $w$), some $\ell_i$ for
the subclause $C'_i\vee\ell_i$ is $(1-\gamma)$-implied by $\neg C'_i$. Thus, 
with probability at least $1-\gamma$, either $\ell_i$ is satisfied in $x$ or 
some other literal of $C'_i$ is satisfied in $x$. We let $\ell'_i$ be the other 
satisfied literal of $C'_i$ in the latter case, and let it be $\ell_i$ in the 
former case. Now, we construct $C_{i+1}$ by removing $\ell'_i$ from $C_i$ and 
taking the next literal of $C$, and we repeat. Note that since $C$ has width at
least $2w$, we find at least $w$ such literals $\ell'_i$. With probability at
least $1-(1-\mu)^w\leq 1-\delta$, at least one of these literals is unmasked.
The first such unmasked literal is as needed.
\end{proof}


The final case to consider is when there is a subclause of width at most $\frac
{1}{\mu\beta}\ln\frac{1}{\delta}$ (noting $\beta\leq 1/2$) for which there is
a literal such that its negation is implied by the negation of the clause. In 
this case, as we elaborate on next, there is a small clause that we can learn 
that can be used to reduce the width of the first clause. If the clause is 
sufficiently wide and yet not witnessed to evaluate to true, it must be because 
there are many such literals. Using these learned clauses to eliminate these 
literals from the given clause enables us to find an equivalent clause that 
satisfies our width bound.

\begin{algorithm}[t!]
\DontPrintSemicolon
\SetKwFunction{WRES}{W-refute}
\SetKwInOut{Input}{input}\SetKwInOut{Output}{output}
\SetKwInOut{Function}{function}

\Function{\WRES($\varphi,w$) decides whether or not there is a width-$w$ 
refutation of CNF $\varphi$.}
\Input{CNF $\varphi$, bound $p(n)$, $\epsilon,\delta,\gamma,\beta\in (0,1)$, 
partial assignments $\rho^{(1)},\ldots,\rho^{(m_0+m_1)}$ from $M(D)$ for
$m_0=\frac{2w(2n+1)^{2w}}{\mu^{2w}\gamma^2}\ln\frac{4(2n+1)}{\delta}$ and
$m_1=\frac{1}{2\gamma^2}\ln\frac{2}{\delta}$ where
$w=\frac{1}{\mu\beta}\ln\frac{2m_1\cdot p(n)}{\delta}$.}
\Output{{\em Accept} or {\em Reject} (cf. Theorem~\ref{mainthm})}

\Begin{
Initialize $\psi$ to an empty CNF.\\
\ForEach{Clause $C$ of width at most $w$}
{
  $\mathit{FALSIFIED}\leftarrow 0$.\\
  \For{$i=1,\ldots,m_0$}
  {
    \If{$C$ is falsified on $\rho^{(i)}$}
    {
      Increment $\mathit{FALSIFIED}$.
    }
  }
  \If{$\mathit{FALSIFIED}\leq\frac{\gamma\mu^w}{2(2n+1)^w}m_0$}
  {
    $\psi\leftarrow\psi\wedge C$.
  }
}
Initialize $\varphi'$ to an empty CNF.\\
\ForEach{Clause $C$ from $\varphi$}
{
 \ForEach{Clause $C'$ from $\psi$}
 {
   \If{$C'=C''\vee\ell$ where $C''\vee\neg\ell$ is a subclause of $C$}
   {
     $C\leftarrow C$ with $\neg\ell$ deleted.
   }
 }
 $\varphi'\leftarrow \varphi'\wedge C$.
}
$\mathit{FAILED}\leftarrow 0$.\\
\For{$i=m_0+1,\ldots,m_0+m_1$}
{
  \If{$\WRES((\varphi'\wedge\psi)|_{\rho^{(i)}},2w)$ rejects}
  {
    Increment $\mathit{FAILED}$.\\
    \If{$\mathit{FAILED}> \lfloor\epsilon\cdot m_1\rfloor$}
    {
      \Return{{\em Reject}}
    }
  }
}
\Return{{\em Accept}}
}

\caption{Learn+RES}\label{mainalg}
\end{algorithm}

\paragraph{Analysis of Algorithm~\ref{mainalg}.}
We first note that the learned clauses are highly valid under $D$ and contain
clauses corresponding to $\ell_1\wedge\cdots\wedge\ell_{k-1}\Rightarrow\ell_k$
whenever $\ell_1\wedge\cdots\wedge\ell_{k-1}\sugg\ell_k$ for $k\leq w$:

\begin{lemma}
\label{learned-valid}
Suppose $D$ is a distribution with a width-$w$ $(\beta,1-\frac{\gamma}
{4(2n+1)^w})$ correlation gap. Let $\psi$ be the conjunction of all clauses of 
width at most $w$ that are witnessed to evaluate to false in at most a $\frac
{\gamma\mu^w}{2(2n+1)^w}$-fraction of a sample of $m_0$ partial assignments from
$M_\mu(D)$ (for $m_0$ as given in Algorithm~\ref{mainalg}). Then with 
probability at least $1-\delta/2$, $\psi$ is $1-\gamma$-valid and contains all 
$1-\frac{\gamma}{4(2n+1)^w}$-valid clauses of width at most $w$, including 
specifically $C=C'\vee\ell$ such that $\neg C'\sugg\ell$.
\end{lemma}
\begin{proof}
Since each clause in $\psi$ has width at most $w$, every literal in each such
clause is simultaneously not set to $*$ by $M_\mu$ with probability at least
$\mu^w$; if a clause is not $1-\frac{\gamma}{(2n+1)^w}$ valid, then it is 
witnessed to evaluate to false on $M_\mu(D)$ with probability at least 
$\frac{\mu^w\gamma}{(2n+1)^w}$. Therefore, Hoeffding's inequality implies that
it is only not witnessed to evaluate to false sufficiently often to eliminate it
from $\psi$ with probability at most $\frac{\delta}{4(2n+1)^w}$. Since there are
fewer than $(2n+1)^w$ clauses of width at most $w$, by a union bound over these 
clauses, we find that with probability at least $1-\delta/4$, these clauses are 
all $1-\frac{\gamma}{(2n+1)^w}$-valid. Again, by a union bound over the clauses,
this means that their conjunction, $\psi$, is $1-\gamma$-valid.

Now, when a clause $C$ of width at most $w$ is $1-\frac{\gamma}{4(2n+1)^w
}$-valid (including $C=C'\vee\ell$ such that $\neg C'\sugg\ell$), Hoeffding's 
inequality similarly guarantees that $C$ will be witnessed to evaluate to false 
on $M_\mu(D)$ in a $\frac{\gamma\mu^w}{2(2n+1)^w}$-fraction of $m_0$ partial 
assignments with probability at most $\frac{\delta}{4(2n+1)^w}$. Therefore, all 
of these clauses appear in $\psi$ except with probability $\delta/4$, and a 
final union bound gives that both conditions hold with probability $1-\delta/2$.
\end{proof}

The analysis of the algorithm is now a generalization of the analysis of the 
width-based algorithm discussed in Section~\ref{unif}. The main twist is that
we may need to modify the resolution derivation by introducing learned clauses
in order to obtain a low-width derivation.

\begin{theorem}\label{mainthm}
Given an input CNF $\varphi$, a bound $p(n)$, $\mu,\epsilon,\gamma,\delta,\beta
\in (0,1)$, and access to examples from $M_\mu(D)$ for a distribution $D$ that 
has a width-$w=\frac{1}{\mu\beta}\ln\frac{2m_1\cdot p(n)}{\delta}$
$\left(\beta,1-\frac{\gamma}{4n(2n+1)^w}\right)$ correlation gap,
(where $m_1=\frac{1}{2\gamma^2}\ln\frac{2}{\delta}$) and 
given that either
\begin{compactitem}
\item $\varphi$ is satisfied by $D$ with probability at least $\epsilon+
2\gamma$ or
\item there exist CNFs $\psi_0$ and $\psi_1$ such that 
\begin{compactitem}
\item $\psi_0$ consists of $1-\frac{\gamma}{4n(2n+1)^w}$-valid clauses and
\item $\psi_1$ is witnessed to evaluate to true with probability 
$1-\epsilon+2\gamma$ under $M_\mu(D)$ 
\end{compactitem}
and there is a resolution refutation of size $p(n)$ of 
$\psi_0\wedge\psi_1\wedge\varphi$
\end{compactitem}
Algorithm~\ref{mainalg} decides which case holds with probability 
$1-\delta$, and runs in time 
$n^{O(\frac{1}{\mu\beta}\log\frac{p(n)}{\gamma\delta})}$.
\end{theorem}
\begin{proof}
We note that the running time bound is essentially immediate from the 
description of the algorithm, as the width-based algorithm, for width $w'$, runs
in time $O(n^{2w'})$. We therefore turn to considering the correctness of the 
algorithm. The first case is similarly simple: by Proposition~\ref
{classical-inf-bound}, if $\varphi$ is $\epsilon+2\gamma$-valid under $D$, then 
since the learned CNF $\psi$ is $1-\gamma$-valid with probability at least $1-
\delta/2$ by Lemma~\ref{learned-valid}, then $\varphi\wedge\psi$ is $\epsilon+
\gamma$-valid---meaning that with probability at least $\epsilon+\gamma$, a 
masked example $\rho$ is drawn from $M_\mu(D)$ for which $(\varphi\wedge
\psi)|_\rho$ is satisfiable. It then follows by Hoeffding's inequality that with
probability greater than $1-\delta/2$, for at least an $\epsilon$-fraction of 
the actual $m_1$ examples, $(\varphi\wedge\psi)|_\rho$ will be satisfiable, and
therefore by the soundness of resolution, at least an $\epsilon$-fraction of the
iterations must fail to find a refutation (of the consequence $(\varphi'\wedge
\psi)|_\rho$), leading the algorithm to reject as needed with probability at 
least $1-\delta$.

It only remains to establish that the algorithm accepts with probability at
least $1-\delta$ when there exists a size-$p(n)$ refutation from some almost
perfectly valid $\psi_0$ and some $\psi_1$ that is witnessed to evaluate to true
with probability $1-\epsilon+2\gamma$ under $M_\mu(D)$. We assume (WLOG) that 
every clause of $\psi_1$ appears in the refutation. We now begin by writing the 
size-$p(n)$ proof $\Pi$ as two subsets of clauses, $\Pi_1$ and $\Pi_0$ where 
$\Pi_1$  consists of clauses which either
\begin{compactenum}
\item contain a subclause $C$ of width $w/2$ (for, recall, $w=\frac{1}{\mu\beta}
\ln\frac{2m_1p(n)}{\delta}$) for which every further subclause $\ell\vee C'$ of 
$C$ has $\ell$ balanced for $\neg C'$ or 
\item contain a subclause of width $w$ for which every (sub-)subclause $C$ of 
width $w/2$ has a further subclause $C'\vee\ell$ for $C'$ and $\ell$ satisfying 
$\neg C'\sugg\ell$,
\end{compactenum}
and $\Pi_0$ consists of the rest of the clauses.

We first note that the clauses in $\Pi_1$ (and in $\psi_1$) are simultaneously
witnessed in most of the partial examples: we first note that by Hoeffding's 
inequality, with probability at least $1-\delta/2p(n)$, each clause in $\psi_1$ 
is witnessed to evaluate to true in at least $(1-\epsilon+\gamma)m_1$ partial 
examples; likewise, for the clauses of $\Pi_1$ which have subclauses for which 
every width $w/2$ subclause has an implied literal, by a union bound over all 
$2n$ literals, the probability that any of them is both a $(1-\frac{\gamma}{4n
(2n+1)^w})$-implied literal indicated by Lemma~\ref{consconsts} and false in any
example is at most $\gamma/4$. (Note that with probability at least $1-\delta/
(2p(n)m_1)$, every such clause has some such unmasked literal in each example.)
Therefore, by Hoeffding's inequality, these clauses are all witnessed
to evaluate to true in at least $(1-\gamma)m_1$ of the examples with probability
at least $1-\delta/(2p(n)m_1)$. Similarly, by Lemma~\ref{wide-witnessed}, the
clauses with $\beta$-balanced subclauses of width at least $w$ are witnessed to
evaluate to true in each partial example with probability at least $1-\delta/
(2p(n)m_1)$. Hence by a union bound over the clauses and examples, every clause 
in $\Pi_1$ (and $\psi_1$) is ultimately (simultaneously) witnessed to evaluate 
to true in at least $(1-\epsilon)m_1$ of the examples considered in the main 
loop of the algorithm with probability $1-\delta/2$.

We now note that for every clause in $\Pi_0$, every subclause of width greater 
than $w$ must have a subclause $C$ of width $w/2$ with some literal $\ell$ such 
that for a further subclause $C'\vee\ell$, $\neg C'\sugg\neg\ell$---every such
$C$ must have a subclause with a literal $\ell$ that is not balanced (or else 
the clause would be in $\Pi_1$ by the first condition) and the further subclause
cannot have $\neg C'\sugg\ell$ for every width $w/2$ subclause $C$ (or else it 
would be in $\Pi_1$ by the second condition). Therefore, by Lemma~\ref
{learned-valid}, the clause $C'\vee\neg\ell$ is added to $\psi$ by the 
algorithm. Moreover, since the clauses of $\psi_0$ are $1-\frac{\gamma}{4n(2n+
1)^w}$-valid, and would reach width $w$ after eliminating fewer than $n$ such 
literals, for the clauses of $\psi_0$ in $\Pi_0$ there is a $\leq w$-CNF 
$\psi'_0$ consisting of subclauses of every clause of $\psi_0$ appearing in 
$\Pi_0$ that (by Proposition~\ref{classical-inf-bound}) are $1-\frac{\gamma}
{4(2n+1)^w}$-valid, and hence also added to $\psi$ by Lemma~\ref{learned-valid}.

Similarly, for every input clause $C$ (from $\varphi$) in $\Pi_0$, we can always
derive a subclause of width at most $w$ using these clauses---as long as our 
subclause of $C$ has width greater than $w$, some subclause of width $w/2$ must 
contain a subclause $C'\vee\ell$ for which a clause $C'\vee\neg\ell$ is in 
$\psi$, which we can use to reduce its size further, so the clauses of 
$\varphi'$ corresponding to clauses of $\varphi$ in $\Pi_0$ have width at most
$w$.

To conclude, we observe that since all of the clauses of $\Pi_1$ (and $\psi_1$)
are witnessed to evaluate to true in at least $(1-\epsilon)m_1$ of the partial 
examples considered by the algorithm in the main loop with probability $1-
\delta/2$, for each such partial example $\rho$, by Proposition~\ref{res-closed}
$\Pi|_\rho$ is a resolution refutation of (restrictions of) the input and 
clauses of $\psi_0$, consisting only of (restrictions of) clauses from $\Pi_0$.
Now, we know that there is a $\leq w$-CNF $\psi'_0$ consisting of subclauses of 
every clause of $\psi_0$ that is included in $\psi$ with probability at least 
$1-\delta/2$, and for each input clause in $\Pi_0$ there is a subclause $C$ of 
width less than $w$ that can be derived by the algorithm. Therefore, by 
induction on the steps of $\Pi|_\rho$, we know that by carrying out the steps of
the proof with the derived subclauses, we can always derive some subclause 
$C'_i$ of each $i$th clause in $\Pi|_\rho$, which is in $\Pi_0$, and hence has 
some further subclause of width at most $w$ which can be derived in width $2w$:
If the next step of $\Pi|_\rho$ is obtained by applying the cut rule to 
$C_i$ and $C_j$, we can apply the cut rule to the derived width-$w$ subclauses 
$C'_i$ and $C'_j$ to obtain a subclause of the next step of width at most $2w$. 
Since this next step is in $\Pi_0$, this clause can be further reduced to width 
$w$ as noted above, completing the induction step.  Since, finally, the only 
subclause of the empty clause (derived in the final step) is the empty clause 
itself, this ultimately yields a width-$2w$ refutation of $\Pi|_\rho$. Since the
algorithm therefore derives the empty clause in at least $(1-\epsilon)m_1$ out 
of the $m_1$ examples with probability at least $1-\delta$, it therefore accepts
with probability at least $1-\delta$ in this case, as needed.
\end{proof}

\paragraph{Remarks on the algorithm.} 
One could divide the derivation steps of our algorithm into ``steps of 
$\Pi|_\rho$'' and ``intermediate derivations,'' where the intermediate 
derivations only involve the width $w$ CNF $\psi$, and serve to find a width-$w$
subclause of the next step of $\Pi|_\rho$. Based on this observation, we could 
have obtained a slightly more complicated algorithm that only uses dynamic 
programming over width-$w$ clauses, that checks after each derivation to see if 
the clause can be reduced to width $w$ by applications of the cut rule to 
clauses from $\psi$ that strictly reduce the size of the intermediate clause
(and discarding the clause if such a reduction is not possible). One might 
prefer it because its complexity is $O(n^{3w})$ rather than $O(n^{4w})$, as 
acheived by our algorithm.

In contrast to the simple algorithm of Section~\ref{unif}
(and more generally, the generic algorithms from prior work~\cite{juba13})
Theorem~\ref{mainthm} also permits the proof to utilize an arbitrary
{\em highly valid} CNF, in addition to one that is {\em witnessed} with
probability $1-\epsilon$. This is accomplished partially due to the simplicity 
of the distribution -- we know that the clauses of any such CNF are either 
witnessed or are essentially equivalent to a short clause -- and partially due 
to the simplicity of the masking process, which allows us to identify all such 
short clauses.

Clearly, our algorithm avoids the sources of intractability indicated by the
various hardness results for the tasks in the two stages. It is easiest to see 
how the problem of Tseitin's hard tautologies~\cite{tseitin70} are avoided in 
the problem we consider: we only certify the validity of a DNF when there exists
a small resolution proof for us to discover, and the proof complexity results 
merely establish that no such small resolution proof exists. So, such an example
is ``out of bounds'' for our algorithm. The way we circumvent difficulties in 
learning is perhaps more interesting. Recall that Valiant's original work in 
this area~\cite{valiant00} concerned algorithms that learned a collection of 
formulas and subsequently reasoned about what these learned formulas entailed. 
Following our previous work~\cite{juba13}, our objective in learning is 
ultimately to determine whether or not {\em there exist} formulas that would 
suffice to complete a proof of a query. Here, we avoid a first common source of 
intractability by not aiming to find CNF representations of our 
hypotheses---essentially as done by ``improper'' learning algorithms. Arguably, 
this is the first step in circumventing the hardness of learning parities 
identified by Michael~\cite{michael10} (but note also that it is still not clear
how to test whether a parity constraint of moderate size is satisfied in our 
affine distributions). Ultimately, the hardness of learning parities in the wRFA
model~\cite{bdd98} is circumvented by our structural result about the resolution
proofs under our distributions: since we can show it suffices to only consider 
logarithmic-width clauses, we effectively only need to learn logarithmic-size 
parities, which is feasible in our partial-information model.

\section{Deciding near-perfect validity of CNF queries}
\label{perfect-cnf}
We now briefly note that our techniques also allow us to strengthen the results
of Khardon and Roth~\cite{kr99} on learning to reason for the special case of
distributions with a width-$O(\log n)$ $(\beta,1-1/q(n))$ correlation gap (such 
as affine distributions) for a quasipolynomial function $q(n)$ and independent 
masking: we show how to decide whether a (general) CNF query is almost perfectly
valid or significantly invalid. In particular, this avoids the requirement that 
the queries have small proofs, but the reader should notice that it applies only
to establishing the validity of CNF (not DNF) queries, only for ``nearly 
perfect'' implications, and only for distinguishing ``nearly perfect'' queries 
from defective ones.

\begin{algorithm}[t!]
\DontPrintSemicolon
\SetKwInOut{Input}{input}\SetKwInOut{Output}{output}

\Input{CNF $\varphi$, $\epsilon,\delta,\gamma,\beta\in (0,1)$, list of partial
assignments $\rho^{(1)},\ldots,\rho^{(m_0+m_1)}$ from $M(D)$ for
$m_0=\frac{32w(2n+1)^{2w}}{\mu^{4w}\gamma^2}\ln\frac{4n+2}{\delta}$ and
$m_1=\frac{32}{(\mu^w\gamma)^2}\ln\frac{2}{\delta}$ where
$w=\frac{2}{\beta}\ln\frac{4|\varphi|}{\gamma}$ and $|\varphi|$ denotes the
number of clauses in $\varphi$.}
\Output{{\em Accept} if $\varphi$ is at least $(1-\frac{\gamma\mu^w}{8})$-valid 
under $D$ or {\em Reject} $\varphi$ is at most $(1-\gamma)$-valid}

\Begin{
Initialize $\psi$ to an empty CNF.\\
\ForEach{Clause $C$ of width at most $w$}
{
  $\mathit{FALSIFIED}\leftarrow 0$.\\
  \For{$i=1,\ldots,m_0$}
  {
    \If{$C$ is falsified on $\rho^{(i)}$}
    {
      Increment $\mathit{FALSIFIED}$.
    }
  }
  \If{$\mathit{FALSIFIED}\leq\frac{\gamma\mu^{2w}}{4(2n+1)^w}m_0$}
  {
    $\psi\leftarrow\psi\wedge C$.
  }
}
Initialize $\varphi'$ to an empty CNF.\\
\ForEach{Clause $C$ from $\varphi$}
{
 \ForEach{Clause $C'$ from $\psi$}
 {
   \If{$C'=C''\vee\ell$ where $C''\vee\neg\ell$ is a subclause of $C$}
   {
     $C\leftarrow C$ with $\neg\ell$ deleted.
   }
 }
 $\varphi'\leftarrow \varphi'\wedge C$.
}
$\mathit{FALSIFIED}\leftarrow 0$.\\
\For{$i=m_0+1,\ldots,m_0+m_1$}
{
  \If{Any clause of $\varphi'$ is falsified on $\rho^{(i)}$}
  {
    Increment $\mathit{FALSIFIED}$.\\
    \If{$\mathit{FALSIFIED}> (5\mu^w\gamma/8) m_1$}
    {
      \Return{{\em Reject}}
    }
  }
}
\Return{{\em Accept}}
}

\caption{CNF-Eval}\label{cnfalg}
\end{algorithm}

\begin{theorem}
Given an input CNF $\varphi$, $\mu,\gamma,\beta\in [0,1]$, and access to 
examples from $M_\mu(D)$ for a distribution $D$ with a width-$w=\frac{2}{\beta}
\ln\frac{4|\varphi|}{\gamma}$ $(\beta,1-\frac{\gamma\mu^w}{8n(2n+1)^w})$ 
correlation gap, and given that either $\varphi$ is at least $(1-\frac
{\gamma\mu^w}{8})$-valid under $D$ or $\varphi$ is at most $(1-\gamma)$-valid, 
Algorithm~\ref{cnfalg} decides which case holds with probability
$1-\delta$ in quasipolynomial time.
\end{theorem}
\begin{proof}
We begin by noting that the running time of the algorithm is indeed 
quasipolynomial in $|\varphi|$, $1/\beta$, $1/\gamma$, $1/\mu$, and $\log 1/
\delta$. Now, towards correctness, we observe that by Lemma~\ref{learned-valid},
the formula $\psi$ consisting of the conjunction of learned clauses is $1-3\mu^w
\gamma/8$ valid with probability $1-\delta/2$. We further note that $\varphi
\wedge\psi\models\varphi'$ since the clauses of $\varphi'$ have a simple 
resolution derivation from clauses of $\varphi$ and $\psi$.

Let us consider the first case, when $\varphi$ is at least $(1-\mu^w\gamma/
8)$-valid. Here, since $\varphi\wedge\psi\models\varphi'$, $\varphi'$ is at
least $(1-\mu^w\gamma/2)$-valid by Proposition~\ref{classical-inf-bound}. 
Therefore, by Hoeffding's inequality, in a sample of size $m_1$ from $D$, it is 
satisfied in at least $(1-5\mu^w\gamma/8) m_1$ of the (underlying, full)
examples with probability at least $1-\delta/2$. Since the partial examples are 
consistent with these full examples, this means in particular that no clause of 
$\varphi'$ can be witnessed to evaluate to false in more than $(5\mu^w\gamma/8) 
m_1$ of the partial examples, so we see that the algorithm accepts with 
probability at least $1-\delta$ in this case, as needed.

Now, turning to the second case, we split the clauses of $\varphi$ into two 
formulas, $\varphi_0$ and $\varphi_1$ as follows: let $\varphi_1$ consist of 
those clauses having either $w/2$ literals that are balanced given any subset of
the rest (of the $w/2$) under $D$ or a subclause of width $w$ such that every 
width $w/2$ subclause has a literal $\ell$ such that for a further subclause $C'
\vee\ell$, $\neg C'\sugg\ell$, and let $\varphi_0$ consist of the rest. Clauses 
in $\varphi_1$ are simultaneously true in most of the partial examples: 
essentially by Lemma~\ref{wide-witnessed} with $\mu=1$ and the definition of 
implied literals, each clause in $\varphi_1$ is true in each underlying example 
(though not necessarily witnessed) with probability at least $1-\gamma/
(4|\varphi|)$. Hence by a union bound over the clauses, every clause in 
$\varphi_1$ is simultaneously true with probability $1-\gamma/4$. By contrast, 
though, since $\varphi$ is at most $1-\gamma$-valid, we see that with 
probability at least $3\gamma/4$, some clause in $\varphi_0$ must evaluate to 
false.

Let $\varphi'_0$ be the clauses of $\varphi'$ corresponding to clauses from 
$\varphi_0$. For every subclause $C$ of a clause of $\varphi_0$ of width greater
than $w$, some further subclause $C'\vee\ell$ of width at most $w/2$ must have 
$\neg C'\sugg\neg\ell$, as if every literal is balanced in any small subclause, 
the orignal clause would have been in $\varphi_1$ by the first condition, and if
some width $w$ subclause $C$ had $\neg C'\sugg\ell$ in some further subclause of
every width $w/2$ subclause of $C$, it would be in $\varphi_1$ by the second 
condition. Much as in the proof of Theorem~\ref{mainthm}, by an analogue of 
Lemma~\ref{learned-valid}, since this $C'\vee\neg\ell$ is $(1-\frac{\gamma\mu^w}
{8(2n+1)^w})$-valid, it is in $\psi$ with probability $1-\delta/2$, and hence 
there is some clause of $\psi$ that can be used to eliminate one of the literals
on this width-$w/2$ subclause. Therefore, we find that the clauses of 
$\varphi'_0$ must have width at most $w$ with probability $1-\delta/2$.

Now, since $M_\mu$ reveals each literal with probability $\mu$ independently,
given that the clauses of $\varphi'_0$ have width at most $w$, each clause in 
$\varphi'_0$ is completely revealed with probability at least $\mu^w$, and hence
some clause in $\varphi'_0$ is witnessed to evaluate to false in the partial 
examples sampled from $M_\mu(D)$ with probability at least $3\mu^w\gamma/4$. 
Therefore, by Hoeffding's inequality, some clause is witnessed to evaluate to 
false in at least a $(5\mu^w\gamma/8)$-fraction of the $m_1$ partial examples 
with probability at least $1-\delta/2$ given that the clauses have width at most
$w$, which we argued occurs with probability at least $1-\delta/2$. Therefore, 
the algorithm rejects with probability at least $1-\delta$ in this case, as
needed.
\end{proof}


\section{Directions for future research}\label{openprobs}
Several problems present themselves as natural directions for future work. The
most pressing of these is, {\em can the restriction to distributions with a
correlation gap be lifted?} That is, how can we efficiently reason about 
``medium-strength'' correlations? Although the ultimate objective of such work
would be to strengthen these results to the distribution-free PAC setting, any
work that handled a class of distributions that exhibited such correlations 
would also be of interest. A similar direction would be to obtain results for a 
more general class of masking processes; although it seems that our results 
generalize to masking distributions that simultaneously reveal any width-$w$ set
of literals with non-negligible probability (for $w=\Omega(\log n)$) such as 
$w$-wise independent distributions (Wigderson and Yehudayoff~\cite{wy12} make a 
similar observation about their algorithm), it would be desirable to find other,
perhaps weaker properties that would also permit relatively efficient 
algorithms.

Of course, the results of this work beg the question so far as the classical 
(quasi)automatizability of resolution is concerned. Although there are families
of counterexamples~\cite{bg01width,ab04} showing that a purely width (and/or 
treelike) based approach to finding small resolution proofs such as pursued by 
Ben-Sasson and Wigderson~\cite{bsw01} cannot beat the current best-known bound 
of $n^{O(\sqrt{n\log n})}$, it does not rule out other approaches. Since our
algorithm and analysis essentially establish that every resolution proof over
distributions with a correlation gap has a low-width approximate version using
the learned clauses, it seems significant for our algorithm that the learned 
formula $\psi$ may not have a small-width derivation. Unfortunately, it is not 
clear how one might hope to exploit this in the absence of a distribution. 
Still, if {\em any} algorithm could find resolution derivations in 
quasipolynomial time, then using the results of our previous work~\cite{juba13},
this would also immediately resolve both of the questions suggested in the 
previous paragraph.


The other natural direction in which one might hope to strengthen our results
involves extending them to stronger proof systems than resolution, such as
cutting planes or polynomial calculus (with resolution, aka PCR). We already
observed in previous work~\cite{juba13} that there are natural fragments of 
these proof systems (most already well studied) for which the combined learning 
and reasoning problem is tractable. The question would be whether, as with 
width-restricted resolution, we could use these algorithms as a starting point 
to obtain algorithms for the unrestricted proof system in the context of
reasoning about a distribution. 

There are, of course, limits to what we can hope for: the strong 
non-automatizability results for systems such as bounded-depth Frege based on 
cryptographic hardness assumptions~\cite{bdgmp04,bpr00,kp95} are established by 
formulas equipped with a natural background distribution (over the variables 
encoding parameters generated by the Diffie-Hellman key exchange protocol~\cite
{dh76}), in which moreover the pattern of masking is independent of the 
underlying example. (The bits are {\em not} masked with the same probability as 
in $M_\mu$, though.) Although these results were not stated in the context of 
PAC-Semantics and the distributions were ignored, they should carry over to an 
appropriate generalization of the setting we considered here, still within the 
framework of the prior work~\cite{juba13}. One might reasonably ask if some
analogous negative result can be proved for $M_\mu$ based on leakage-resilient 
cryptography (as $M_\mu$ leaks a constant fraction of the secret parameters).

\section*{Acknowledgements}
The author would like to thank Paul Beame, Eli Ben-Sasson, and Leslie Valiant 
for comments and conversations that helped shape this work.

\bibliographystyle{plain}
\bibliography{robust}

\end{document}